\pdfoutput=1
\documentclass[%
 reprint,
 floatfix,
 amsmath,
 amssymb,
 aps,
 nofootinbib,
 groupedaddress
]{revtex4-2}

\usepackage{graphicx}%
\usepackage{quantikz}
\usepackage{braket}
\usepackage[margin=1.2in,top=1.2in,bottom=1.2in]{geometry}
\usepackage{amsthm}
\usepackage{thmtools}
\usepackage{xcolor}
\usepackage[
    colorlinks=true,
    linkcolor=blue!70!black,
    citecolor=blue!70!black,
    urlcolor=red!70!black
]{hyperref}
\makeatletter
\newwrite\app@aux
\immediate\openout\app@aux=\jobname-app.aux
\immediate\write\app@aux{\relax}
\immediate\write\app@aux{\string\bibstyle{coherent_consensus_bibliographicstyle}}
\immediate\write\app@aux{\string\bibdata{coherent_consensus_references}}
\DeclareRobustCommand{\citeapp}[1]{%
  \immediate\write\app@aux{\string\citation{#1}}%
  {\@fileswfalse\cite{#1}}}
\makeatother

\newtheorem{corollary}{Corollary}
\newtheorem{fact}{Fact}
\newtheorem{definition}{Definition}
\newtheorem{remark}{Remark}

\usepackage{subcaption}
\DeclareCaptionJustification{flushing}{\flushing}
\usepackage[nameinlink,capitalise]{cleveref}
\crefname{theorem}{Theorem}{Theorems}
\Crefname{theorem}{Theorem}{Theorems}
\crefname{corollary}{Corollary}{Corollaries}
\Crefname{corollary}{Corollary}{Corollaries}
\crefname{fact}{Fact}{Facts}
\Crefname{fact}{Fact}{Facts}
\crefname{definition}{Definition}{Definitions}
\Crefname{definition}{Definition}{Definitions}
\crefname{appendix}{Appendix}{Appendixes}
\Crefname{appendix}{Appendix}{Appendixes}
\crefname{equation}{Eq.}{Eqs.}
\Crefname{equation}{Equation}{Equations}
\crefname{figure}{Fig.}{Figs.}
\Crefname{figure}{Figure}{Figures}
\crefname{section}{Sec.}{Secs.}
\Crefname{section}{Section}{Sections}

\makeatletter

\def\frontmatter@above@affiliation{\vskip 6\p@\relax}

\AtBeginDocument{%
  \long\def\@makecaption#1#2{%
    \par
    \vskip\abovecaptionskip
    \begingroup
      \small\rmfamily
      \samepage
      \flushing
      \parindent\z@
      \let\footnote\@footnotemark@gobble
      \@make@capt@title{#1}{#2}\par
    \endgroup
    \vskip\belowcaptionskip
  }%
}

\makeatother
\DeclareMathOperator{\e}{e}
\DeclareMathOperator{\im}{i\!}
\DeclareMathOperator{\Var}{Var}
\DeclareMathOperator{\uniform}{U}
\DeclareMathOperator{\upperO}{O}
\DeclareMathOperator{\mv}{MV}
\DeclareMathOperator{\cc}{CC}
\DeclareMathOperator{\total}{total}
\DeclareMathOperator{\stat}{stat}
\DeclareMathOperator{\sys}{sys}
\DeclareMathOperator{\theory}{th}
\DeclareMathOperator*{\medieumotimes}{\text{\raisebox{0.25ex}{\scalebox{0.8}{$\bigotimes$}}}}

\newcommand{\blk}[1]{\ensuremath{\vphantom{U(\phi)B_j^\dagger}#1}}
\newcommand{\Had}{\gate[disable auto height,/tikz/inner xsep=0pt,/tikz/minimum width=10pt]{\blk{H}}}
\newcommand{\M}{\gate[disable auto height,/tikz/inner xsep=0pt,/tikz/minimum width=10pt]{\blk{M}}}
\def\Uj[#1]{\gate[disable auto height,/tikz/inner xsep=0pt,/tikz/minimum width=10pt]{\blk{U_{#1}}}}
\newcommand{\gatedots}{\gate[disable auto height,/tikz/inner xsep=0pt,/tikz/minimum width=10pt]{\blk{\dots}}}
\newcommand{\vlab}[1]{\rotatebox{270}{\smash{\ensuremath{#1}}\vphantom{\ensuremath{\text{GHZ}}}}}
\newcommand{\GHZ}{\gate[wires=4,disable auto height,/tikz/inner ysep=0pt,/tikz/minimum width=10pt]{\vlab{\text{GHZ}}}}
\newcommand{\GHZdagger}{\gate[wires=4,disable auto height,/tikz/inner ysep=0pt,/tikz/minimum width=10pt]{\vlab{\text{GHZ}^\dagger}}}

\makeatletter
\def\@author@present@script#1#2#3{%
 \begingroup
  \gdef\comma@space{\textsuperscript{\,}}%
  \doauthor{#2}{#3}{\@affil@present@script}%
 \endgroup
 \advance\@tempcnta\m@ne
}%
\makeatother

\makeatletter
\let\auto@bib\@empty
\let\auto@bib@innerbib\@empty
\makeatother
\begin{document}

\title{Coherent consensus for the estimation and classification of phases}

\author{Roberto Ruiz}
\author{Alessandro Luongo}
\author{Sergi Ramos-Calderer}
\author{Jos\'e Ignacio Latorre}

\affiliation{Centre for Quantum Technologies (CQT), National University of Singapore 117543, Singapore}

\begin{abstract}
	\noindent
	Inspired by quantum metrology,
	we propose a
		{\sl coherent consensus}
	primitive to enhance the collaborative performance of independent quantum classifiers at the decision boundary.
	The basic idea consists of three steps:
	i) prepare a Greenberger-Horne-Zeilinger (GHZ) state on $n$ qubits;
	ii) apply a relative phase to each qubit, which encodes the output of one classifier;
	and iii) undo the GHZ preparation to collect all phases onto the first qubit.
	The signal is amplified by a factor of $n$,
	which reduces the sampling error by a factor of $\sqrt{n}$ relative to incoherent majority voting,
	that is, the direct combination of estimates from individual classifiers.
	Coherent consensus is valid for general estimation and classification of phases.
	We derive guarantees on the resources needed to estimate phases using this coherent consensus primitive 
	as compared to majority voting.
	We implement both techniques on a trapped-ion quantum computer
	using a given ensemble of imperfect phase gates.
	We observe the expected $\sqrt{n}$ gain in precision of coherent consensus,
	with excess dispersion at the largest $n$ consistent with device drift.
	Coherent consensus can
	thus be understood as a form of sensing inside a quantum computer.

\end{abstract}

\maketitle

\section{Introduction}
\label{sec_intro}

\noindent Efficiently learning information encoded in a quantum process is central to algorithm design.
When this information is accessible only through a black box,
the learning cost is quantified by the minimum number of oracle calls needed to reach a prescribed error and success probability.
Quantum metrology~\cite{Bollinger96, Giovannetti06}
poses an analogous question when the information to be learned is an unknown phase~$\phi$ encoded in a unitary $U_\phi$ acting as \begin{equation}
	\label{Uphi}
	\begin{split}
		U_\phi & \frac{1}{\sqrt 2}\left(|\!\uparrow\rangle + |\!\downarrow\rangle\right) \\= &\frac{1}{\sqrt 2}(\e^{-\im \phi/2}|\!\uparrow\rangle + \e^{\im \phi/2}|\!\downarrow\rangle) \ ,
	\end{split}
\end{equation}
where this phase can be produced by a tiny external magnetic field $\vec B=B \hat{z}$ along the $z$ axis.
To be precise, 
for an evolution $U_\phi=\e^{-\im T H}$ with Hamiltonian $H=-\vec\mu\cdot \vec B$, 
where $\vec\mu=\mu\vec\sigma/2$ is the spin magnetic moment and $\hbar=1$, 
the phase accumulated during the rotation is
$\phi=-T\mu B$~\cite{Bollinger96}.
The learning problem is to estimate this magnetic field,
whose action can be understood as an oracle call to this phase rotation.

In the setting of quantum metrology, 
entanglement provides the means for a quadratic improvement in the precision scaling.
Let us recall how to achieve this improvement. 
The $n$ probes
(or qubits)
are first entangled in a Greenberger-Horne-Zeilinger (GHZ) state.
The magnetic field adds a phase to each qubit.
The phases are collected onto a single qubit by undoing the GHZ preparation.
The net effect is an $n$-fold amplification of the signal $\phi$. 
A Hadamard gate is then applied to the first qubit to read out the result. 
The final probability of obtaining $\downarrow$ in the $X$ basis is $\Pr(\downarrow)=\sin^2(n \phi/2)\simeq n^2\phi^2/4$, which scales quadratically with $n$
for small signals.
This quadratic improvement is known as the \emph{Heisenberg limit}.
If no entanglement were used,
each of the $n$ experiments would read a signal $\Pr({\downarrow})=\sin^2(\phi/2)\simeq\phi^2/4$, 
which would then be classically combined to produce $n \phi^2/4$. This scales linearly with $n$, corresponding to the \emph{standard quantum limit}.
Therefore, entanglement enhances the signal by a factor of $n$.
For small~$\phi$,
this reduces the sampling effort needed to reach a specified precision by a factor of $n$.

We call 
the classical strategy that estimates the phase at the standard quantum limit \textit{majority voting},
and the strategy that uses entanglement and interference to reach the Heisenberg limit
\textit{coherent consensus}.

We note that the setup where a phase $\phi$
is encoded in a single qubit by the action of a matrix $U_\phi$ effectively arises in several contexts:
phase kickbacks~\cite{Cleve98},
calibrations of single-qubit rotations~\cite{Kimmel15},
sensing networks~\cite{Komar14, Proctor18, Eldredge18, Rubio20, Liu21},
or,
under quantum signal processing~\cite{Low17},
amplitude estimation~\cite{Rall23, Oshio25}.

\begin{figure*}[t]
    \centering
    \includegraphics[width=.48\linewidth]{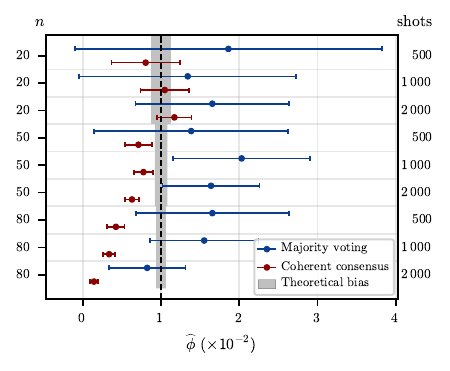}
    \hfill
    \includegraphics[width=.48\linewidth]{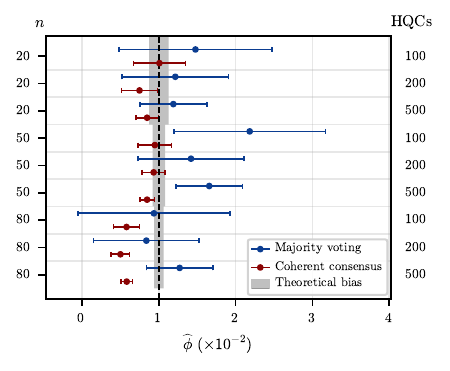}
\caption{Estimated phases
		via majority voting and coherent consensus
		with imperfect phase gates
		on the trapped-ion quantum computer~\cite{Ransford25}.
		The left panel contains the
		results at equal shot budgets, while the right panel shows outcomes at equal resource budgets, 
		quantified as hardware quantum credits (HQCs)~\cite{HQCs}.
		The dashed line corresponds to the target
		$\phi = 10^{-2}$.
		The markers represent the estimates $\widehat\phi$, and
		the error bars represent the statistical error at $95\%$ due to shot noise.
		Each imperfect phase gate carries an angle $\phi_j$, 
		drawn uniformly from a band of half-width $a=5\times10^{-3}$ around $\phi$.
		The gray shaded area represents
		the range, narrowing as $n$ grows,
		within which the theoretical bias of the estimator due to $\phi_j$
		falls for $95\%$ of the draws.
		Majority voting has wider error bars that capture $\phi$ more frequently
		($7/9$ left, $7/9$ right).
		The error bars of coherent consensus
		shrink faster
		due to its
		improved precision,
		but capture
		the target less often
		($3/9$ left, $3/9$ right).
		The theoretical bias is too small to account for the departures of coherent consensus at $n=50,80$.
		This is a sign of systematic errors resulting in a loss of accuracy.
		Yet coherent consensus 
		is still more accurate 
		than majority voting for $n=20,50$ 
		in most cases 
		($6/6$ left, $5/6$ right).
	}
	\label{fig_exp_coverage}
\end{figure*}

\begin{figure}[ht]
	\centering
	\begin{subfigure}[b]{\linewidth}
		{\centering
			\begin{quantikz}[row sep={0.575cm,between origins}, column sep=0.3cm]
				\lstick{$\ket{0}$} & \Had & \Uj[1]    & \M & \meter{} \\
				\lstick{$\ket{0}$} & \Had & \Uj[2]    & \M & \meter{} \\
				\lstick{$\ket{0}$} & \Had & \gatedots & \gatedots    & \meter{} \\
				\lstick{$\ket{0}$} & \Had & \Uj[n]    & \M & \meter{} \\
			\end{quantikz}\par}
		\caption{Multiqubit circuit realizing majority voting.
			The gates $U_{\phi_j}$ implement
			the target phase $\phi$
			imperfectly.
			The gate $M$ implements
			the
			change to the measurement basis.
			The phase $\phi$ is reconstructed
			by aggregating measurement outcomes.}
		\label{fig_maj_qc}
	\end{subfigure}
	\vspace{0.3cm}
	\begin{subfigure}[b]{\linewidth}
		{\centering
			\begin{quantikz}[row sep={0.575cm,between origins}, column sep=0.3cm]
				\lstick{$\ket{0}$} & \Had & \GHZ & \Uj[1]    & \GHZdagger & \M & \meter{} \\
				\lstick{$\ket{0}$} &           &      & \Uj[2]    &            &          &
				\\
				\lstick{$\ket{0}$} &           &      & \gatedots &            &          &
				\\
				\lstick{$\ket{0}$} &           &      & \Uj[n]    &            &          &
				\\
			\end{quantikz}\par}
		\caption{Multiqubit circuit realizing coherent consensus.
			The gate $M$ implements
			the change to the measurement basis.
			The GHZ state
			allows the phases $\phi_j$ to accumulate on the first qubit, from which
			$\phi$ is estimated.}
		\label{fig_coh_qc}
	\end{subfigure}
	\caption{Circuits to estimate $\phi$ with
		majority voting (\cref{fig_maj_qc}) and coherent consensus (\cref{fig_coh_qc}).}
	\label{fig_circuits}
\end{figure}
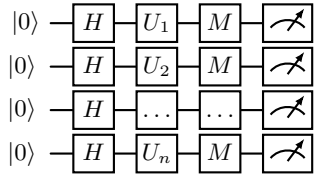
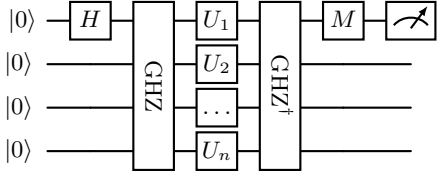

In this work,
we use coherent consensus to learn from an ensemble of imperfect binary classifiers with a finite shot budget.
Recall that a trained classifier sorts unseen data into categories,
yet items near the decision boundary are hard to distinguish,
requiring a larger shot budget.
Moreover, 
independently trained instances may disagree,
which can blur the boundary further.
We encode the binary classification in the sign of the phase $\phi$ of a phase gate $U_\phi$.
Ensemble disagreement is thus represented by~$n$ slightly imperfect
phases $\phi_j$ in place of the target~$\phi$,
where~$n$ is the number of qubits.
For an ensemble with phases $\phi_j$,
determining
the sign of~$\phi$
with a finite shot budget
has limited precision due to statistical sampling error.
In an experimental setting, 
the measurement results will carry further systematic errors affecting accuracy.
We derive sufficient conditions on the number of qubits~$n$
and the number of shots needed
to reach a prescribed relative error in the estimation of $\phi$ from imperfect phase gates for both majority voting~(\Cref{thm:relMV}) and coherent consensus~(\Cref{thm:relCC}).
With these results,
we give a decision guarantee for binary classification~(\Cref{cor:sign}).
Finally, we demonstrate that
coherent consensus outperforms majority voting
in precision, and in accuracy for $n=20,50$, on a trapped-ion quantum computer~\cite{Ransford25}, 
as illustrated in \cref{fig_exp_coverage}.

\subsection*{Previous work}

\noindent Phase estimation with
entangled probes
goes back to seminal works
on quantum
metrology~\mbox{\cite{Bollinger96, Giovannetti06}}.
If the unknown parameter
is an amplitude
rather than a phase,
the Heisenberg limit is
reached through
amplitude estimation~\mbox{\cite{Brassard02, Knill07}}.
Parallel
amplitude estimation~\cite{Oshio25}
can trade depth for width, using techniques very similar to those of quantum metrology.
There,
quantum signal processing~\cite{Low17}
turns the
oracle into
a rotation of the form given in \cref{Uphi}.
The phases then accumulate
in the relative phase of a GHZ state,
achieving a depth of
$\upperO(\log(1/\eta))$
at the cost of
$\upperO(\eta^{-1}\log(1/\eta))$ queries, where $\eta$ is
the target error in the amplitude.
In this work, we instead provide new bounds on the total number of
qubits $n$ and the number of
shots required
for phase estimation with imperfect phases $\phi_j$
to achieve a prescribed relative error with a given success probability.

Coherent consensus demands an upper bound on the phase
to avoid branch-cut ambiguities.
This can
be circumvented by
techniques inspired by
robust phase estimation~\cite{Higgins09, Kimmel15},
where shots are allocated across an increasing number of applications of the unitaries.
For entangling protocols,
the GHZ state is replaced with a product of GHZ states of varying sizes~\cite{Direkci26}.
Here we consider $n$ unitaries at a predetermined scale,
for which the bounds for coherent consensus hold by assumption, limiting admissible $n$ from above.

In sensing networks,
entanglement between sensors does not yield a gain for independent local phases,
but it is effective for global properties of the network.
Coherent consensus uses entanglement to estimate
the global, equally weighted
linear function of local
phases~\cite{Proctor18, Eldredge18}.
Coherent consensus reduces the variance
relative to majority voting by a factor of
$n$~\cite{Eldredge18}.
If the phases are perfect,
the precision can be bounded analytically using the quantum
Fisher information~\cite{Proctor18, Eldredge18},
while, for two sensors,
the mean-square error can be evaluated numerically~\cite{Rubio20}.
In fact,
finite-sample bounds on the mean-square error
have been derived for phase estimation with maximally entangled probes~\cite{Belliardo20}. 
Distributed estimation of an averaged phase with entangled probes has been demonstrated experimentally~\cite{Liu21}.
Programmable metrology using low-depth circuits has also been demonstrated on a trapped-ion quantum processor~\cite{Marciniak22}.
Entanglement-assisted supervised classification was proposed for sensor networks~\cite{Zhuang19} 
and demonstrated experimentally~\cite{Xia21}.

Prior work also includes \textit{quantum ensembles of classifiers}~\cite{Schuld18} and \textit{quantum bagging constructions}~\cite{Macaluso24}. 
Adaptive boosting with quantum classifiers has been demonstrated on a programmable superconducting processor~\cite{Chen25}. 
These approaches differ from the present setting in their oracle assumptions and aggregation procedures. 
Here, we study coherent phase accumulation with imperfect phase gates, 
derive finite-shot error guarantees, 
and test the resulting primitive on trapped-ion hardware.

Finally, the related term \textit{quantum majority vote}~\cite{Buhrman23} 
concerns approximating the more frequently occurring state in a product of qubits, 
each prepared in one of two orthogonal states in an unknown basis. 
In the present work, majority voting instead denotes the classical aggregation of measurement outcomes for phase estimation.

\section{Coherent consensus and majority voting}
\label{sec_theory}

\noindent We are interested in
the
advantage of coherent consensus
over
majority voting
as a primitive in digital quantum computing. We first set up the appropriate quantum circuits for each strategy. We then analyze the errors that arise in
experimental deployment. Finally, we
establish theoretical learning success guarantees.

\subsection*{Quantum circuits}

\noindent
Let us construct the quantum circuits for each strategy.
The first step is to model the unitary
$U_{\phi}$ in \cref{Uphi}
as a single phase gate.
This unitary $U_{\phi}$ could encode the results of a single-qubit classifier~\cite{PerezSalinas19, PerezSalinas21}, 
or those of a complex multiqubit process that is distilled into a phase rotation.

The task of both coherent consensus
and incoherent majority voting is to estimate $\phi$ and
classify its sign using an ensemble
of \mbox{$j=1,\ldots,n$} phase gates $U_{\phi_j}$.
We consider a model in which $U_{\phi_j}$ do not realize $\phi$~perfectly,
but instead yield different assessments $\phi_j$ of $\phi$,
\begin{equation}
	\begin{split}
		\label{U}
		U_{\phi_j}|+\rangle
		=
		\frac{1}{\sqrt{2}}(|0\rangle
		+\e^{\im\phi_j}|1\rangle) \ .
	\end{split}
\end{equation}
The deviations are expected to be mild if $U_{\phi_j}$ are well trained.
Small phases $\phi$, with which we are concerned, are the hardest to classify.

Current quantum hardware has access to fairly large qubit budgets,
yet the feasible circuit depth remains limited.
Therefore, we seek circuit implementations of the phase-estimation strategies that minimize depth at the expense of width.
The quantum circuits for majority voting and coherent consensus
are displayed in \cref{fig_maj_qc,fig_coh_qc}, respectively.
Each strategy leads to an estimator $\widehat{\phi}$ of $\phi$. In the plots, we use the same notation for both the estimator and the estimate.

The majority voting circuit is presented in \cref{fig_maj_qc}. It
has constant
depth, as the phase gates are arranged in parallel. The circuit implementing majority voting is
\begin{equation}
	U_{\mv} =
	M^{\otimes n}
	\Bigl(\underset{j=1}{\overset{n}{\medieumotimes}}\,U_{\phi_j}\Bigr)
	H^{\otimes n} \ ,
\end{equation}
where $M$ is the change of basis, namely, $M=H$ for measurements in the \mbox{$X$ basis} and \mbox{$M=HS^\dagger$} for measurements in the \mbox{$Y$ basis}. Note that
\begin{equation}
	\!\!\!\!
	\Bigl(\underset{j=1}{\overset{n}{\medieumotimes}}\,U_{\phi_j}\Bigr)
	H^{\otimes n}\ket{0}\!\vphantom{H}^{\otimes n}
	= \underset{j=1}{\overset{n}{\medieumotimes}}\frac{1}{\sqrt 2}(|0\rangle + \e^{\im \phi_j}|1\rangle) \ .
\end{equation}
Estimates of the bit frequencies
of each qubit
yield $n$ different estimators $\widehat{\phi}_j$.
The estimator of majority voting can be thought of as the average of different estimates
\footnote{To prove~\Cref{thm:relMV,thm:relCC} 
in \Cref{app_theorems}, 
we rather collect the measurement statistics for each qubit and average them to obtain $\widehat{\phi}$.
Predictions for small $\phi$ are in practice the same either way.}
\begin{equation}
	\label{phimaj}
	\widehat{\phi}_{\mv}:=\frac{1}{n}\sum_{j=1}^n\widehat{\phi}_j \ .
\end{equation}
Majority voting is thus the aggregate of individual qubit estimators.

The coherent consensus circuit in \cref{fig_coh_qc} prepares the GHZ state, applies the tensor product of
phase gates, and undoes the preparation.
The circuit implementing coherent consensus is
\begin{equation}
	U_{\cc}=M_1
	U_{\text{GHZ}}^\dagger
	\Bigl(\underset{j=1}{\overset{n}{\medieumotimes}}\,U_{\phi_j}\Bigr)
	U_{\text{GHZ}}^{\vphantom{\dagger}}\,H_1 \ ,
\end{equation}
where $H_1$ is a Hadamard gate on the first qubit,
$U_{\text{GHZ}}$ is an entangling subcircuit preparing the GHZ state, 
and $M_1$ is the change of basis on the first qubit,
the only one that is measured.
Note that
\begin{equation}
	\begin{split}
		  & \,U_{\text{GHZ}}^\dagger
		\Bigl(\underset{j=1}{\overset{n}{\medieumotimes}}\,U_{\phi_j}\Bigr)
		U_{\text{GHZ}}^{\vphantom{\dagger}}
		H_{1}\ket{0}\!\vphantom{H}^{\otimes n}                                                                     \\
		= & \,U_{\text{GHZ}}^\dagger
		\Bigl(\underset{j=1}{\overset{n}{\medieumotimes}}\,U_{\phi_j}\Bigr)\frac{1}{\sqrt{2}}(\ket{0}\!\vphantom{H}^{\otimes n}
		+\ket{1}\!\vphantom{H}^{\otimes n})                                              \\
		= & \,U_{\text{GHZ}}^\dagger\frac{1}{\sqrt{2}}(\ket{0}\!\vphantom{H}^{\otimes n}
		+\e^{\im\sum\limits_{j=1}^{n}\!\phi_j}\!
		\ket{1}\!\vphantom{H}^{\otimes n})                                               \\
		= & \,\frac{1}{\sqrt{2}}(\ket{0}
		+
		\e^{\im\sum\limits_{j=1}^{n}\!\phi_j}\!
		\ket{1})\ket{0}\!\vphantom{H}^{\otimes n-1}\ .
	\end{split}
\end{equation}
Measuring the
first qubit allows us to estimate the sum of phases
from the bit statistics.
The estimator of coherent consensus is
the normalized estimator of the sum,
\begin{equation}
	\label{phiscoh}
	\widehat{\phi}_{\cc}:=\frac{1}{n}\widehat{\sum_{j=1}^n \phi_j} \ .
\end{equation}
Note that the estimator
is single valued as long as the sum stays within $(-\pi,\pi]$ when measurements are performed in both the $X$ and $Y$ bases,
and within $[-\pi/2,\pi/2]$ when measurements are performed only in the $Y$ basis.

The quantum resources required for coherent consensus depend on the compilation of the GHZ circuit $U_{\text{GHZ}}$.
This circuit carries an unavoidable overhead of two-qubit gates over majority voting,
reflecting the cost of entanglement generation.
Moreover, a circuit with logarithmic depth in $n$ can be obtained by exploiting 
all-to-all connectivity using a number of two-qubit gates linear in $n$.
Furthermore, when ancillas and feedforward operations are available, a circuit of constant depth can be obtained,
with a number of two-qubit gates linear in $n$~\cite{Buhrman23ii}.

\subsection*{Accuracy versus precision}

\noindent
The error of the estimators $\widehat{\phi}$ of majority voting and coherent consensus has three contributions. 
First, statistical errors due to the shot sampling of the circuit. 
Second, systematic errors arising from imperfect circuit implementation, such as quantum-device noise.
Third, theoretical errors, which in our case are training deviations of~$\phi_j$ from $\phi$.
If the three contributions
are independent,
the error of $\widehat{\phi}$ about $\phi$, 
taken as
the root mean square deviation,
decomposes as
\begin{equation}
	\label{total_error}
	\varepsilon_{\total}
	=
	\sqrt{\varepsilon_{\stat}^2
		+\varepsilon_{\sys}^2
		+\varepsilon_{\theory}^2}
	\ .
\end{equation}

The difference between the statistical error and the other two is worth stressing.
The accuracy is quantified by the difference between the mean of the estimator and the ground truth, that is, its bias.
An accurate estimator has no bias and targets $\phi$, but may produce widely dispersed estimates.
Both systematic errors and theoretical errors can bias the estimator and thereby affect its accuracy.
In our framework, systematic errors arise from the physical implementation on a quantum device,
and theoretical errors encompass finite training effects and limitations of the training model. 

By contrast, statistical errors limit the precision of the estimator $\widehat{\phi}$. 
The precision quantifies the dispersion of estimates around their mean.
A precise estimator consistently produces estimates lying within a narrow interval, 
but this interval may fail to include $\phi$ if the estimator is biased.
In our framework, 
statistical error for majority voting and coherent consensus is mainly due to
sampling error,
although
other sources of error could in principle contribute to dispersion as well.
This assumes that the device noise is uncorrelated across shots.
Majority voting obeys the standard quantum limit,
\begin{equation}
	\label{stat_var_mv}\varepsilon_{\mv,\stat}=\sqrt{\Var_{\text{shot}}[\widehat\phi_{\mv}]}
	=
	\upperO\!\left(
	\frac{1}{\sqrt{nm}}
	\right) \ ,
\end{equation}
where $m$ denotes the number of shots, whereas coherent consensus obeys the Heisenberg limit,
\begin{equation}
	\label{stat_var_coh}
	\varepsilon_{\cc,\stat}=\sqrt{\Var_{\text{shot}}[\widehat\phi_{\cc}]}
	=
	\upperO\!\left(
	\frac{1}{n \sqrt{m}}
	\right) \ .
\end{equation}
Thus, coherent consensus offers a precision gain over majority voting by a factor of 
$\sqrt{n}$, 
a statement that we now formalize.

\subsection*{Learning and theoretical guarantees}

\noindent
Beyond its roots in quantum metrology, the classification of~$\phi$ admits a neat formulation akin to the query model of computation~\cite{Angluin88, bennett1997strengths}.
Accordingly, the phase~$\phi$ corresponds to the optimal Bayes score, the hypothesis class is the family $\{U_{\phi^\prime}\}$, and the ensemble members $U_{\phi_j}$  are black boxes whose phases are unknown to the learner.
This representation inspires us to derive
success guarantees for both majority voting and coherent consensus. 
To obtain them, however, we need some assumptions.

To mimic the result of training on the phases,
we assume the decomposition $\phi_j = \phi + \gamma_j$ in \cref{U}, where the perturbations~$\gamma_j$ 
are independent and identically distributed (i.i.d.).
These perturbations are fixed across runs of the circuit.
In the limit of a large number of qubits,
we expect the resulting phases to exhibit bounded noise around~$\phi$.
We also expect the average of the perturbations $\gamma_j$ to approach zero as $n$ grows. 
The definition of imperfect phase gates formalizes these assumptions.

\begin{definition}[Imperfect phase gates]
	\label{def:imperfect}
	Consider the phase $\phi\in(-\pi,\pi]$ and let $n\in\mathbb N$.
	An \emph{ensemble of $n$ imperfect phase gates} for $\phi$ is the family of
	unitaries $U_{\phi_j}$ in \cref{U}, $j=1,\dots,n$, with phases
	\begin{equation}
		\label{phij}
		\phi_j = \phi + \gamma_j \ ,
	\end{equation}
	where the perturbations are i.i.d. random variables
	bounded almost surely by some $a\ge 0$ and symmetric about zero:
	\begin{equation}
		\label{gammaj}
		\Var[\gamma_j]=\sigma^2 \ , \quad
		|\gamma_j|\le a \ , \quad
		\gamma_j\overset{\mathrm d}{=}-\gamma_j \ .
	\end{equation}
\end{definition}

Note that, while the definition does not strictly encompass Gaussian noise $\gamma_j$ since $|\gamma_j|\le a$ almost surely, the results are expected to hold approximately if the tails of $\gamma_j$ beyond $a$ are negligible.
Moreover,
the formalization of imperfect phase gates in \Cref{def:imperfect} also enables us to compute the scaling of theoretical errors.
Once sampled, the $\gamma_j$ remain fixed throughout different runs of the circuits. For small symmetric perturbations $\gamma_j$, it follows that
\begin{equation}
	\label{th_var}
	\begin{split}
		\varepsilon_{\mv,\theory}
		=\varepsilon_{\cc,\theory}
		=
		\upperO\!\left(
		\frac{\sigma}{\sqrt{n}}
		\right) \ .
	\end{split}
\end{equation}
Unlike statistical errors in \cref{stat_var_mv,stat_var_coh}, the theoretical error of collecting $n$ imperfect phase gates is the same for majority voting and coherent consensus. 
Since the $\gamma_j$ are held fixed, 
theoretical error is just a bias.
Only across renewed draws of the perturbations $\gamma_j$
does the bias offset have zero mean and spread $\sigma/\sqrt{n}$.
If $\sigma\ll 1/\sqrt{m}$,
the bias is subleading with respect to shot dispersion for majority voting;
if $\sigma\ll 1/\sqrt{mn}$ instead,
it is subleading for coherent consensus.

We now consider success guarantees for the estimator $\widehat{\phi}$. The statements require informed bounds on the otherwise unknown $|\phi|$. This situation corresponds to a controlled-classification scenario. The lower bound $\phi_{\min}>0$ prevents the target phase from approaching zero arbitrarily closely. 
Conversely, the upper bound~$\phi_{\max}$ prevents $\phi$ from wrapping around the unit circle. In coherent consensus, requiring the estimator to remain within $(-\pi,\pi]$ implies that~$\phi_{\max}$ tightens as $n$ grows.

\begin{restatable}[Bound on the relative error for majority voting]{theorem}{relativeMVlearning}
	\label{thm:relMV}
	Let
	$\phi$ be a given phase such that
	$\phi_{\min}\le|\phi|\le\phi_{\max}$,
	with $\phi_{\min}>0$.
	Let
	$\phi_j=\phi+\gamma_j$
	be the
	phases
	of $n$
	imperfect
	phase gates (\Cref{def:imperfect})
	in the quantum circuit of \cref{fig_maj_qc}.
	Denote
	the relative error by $\epsilon$~$\in (0,1)$
	and
	the
	failure probability by
	$\delta\in (0,1)$,
	and assume
	\begin{equation*}
		0\le a<\frac{\pi}{2}
		\quad\text{and}\quad
		\phi_{\max}+\epsilon\phi_{\min}<\pi \ .
	\end{equation*}
	Furthermore,
	let the number of qubits $n$ and the even number of shots $m$
	satisfy
	\begin{equation*}
		\begin{split}
			n & \ge
			\left(
			\frac{8\sigma^2}{\cos^2 a\,\epsilon^2\phi_{\min}^2}
			+
			\frac{4\tan a}{3\epsilon\phi_{\min}}
			\right)\ln\frac{4}{\delta} \ ,
			\\
			m & \ge
			\frac{128}{n\cos^2 a\,\epsilon^2\phi_{\min}^2}
			\ln\frac{8}{\delta} \ .
		\end{split}
	\end{equation*}
	Then, there is an estimator using the measurement statistics from \cref{fig_maj_qc} that satisfies
	\(
	\Pr(|1-\widehat{\phi}_{\mv}/\phi|\le \epsilon) \ge 1-\delta \ .
	\)
\end{restatable}

\begin{restatable}[Bound on the relative error for coherent consensus]{theorem}{relativeCClearning}
	\label{thm:relCC}
	Let
	$\phi$ be a given phase such that
	$\phi_{\min}\le|\phi|\le\phi_{\max}$,
	with $\phi_{\min}>0$.
	Let $\phi_j=\phi+\gamma_j$
	be the phases of $n$
	imperfect phase gates (\Cref{def:imperfect})
	in the quantum circuit of \cref{fig_coh_qc}.
	Denote the relative error by $\epsilon$~$\in (0,1)$ and the
	failure probability by
	$\delta\in (0,1)$,
	and assume
	\begin{equation*}
		\phi_{\max}<\frac{\pi}{n(1+\epsilon)}-a \ ,
	\end{equation*}
	where $n$ is the number of qubits.
	Let $n$ and the even number of
	shots $m$ satisfy
	\begin{equation*}
		\begin{split}
			n & \ge
			\left(
			\frac{8\sigma^2}{\epsilon^2\phi_{\min}^2}
			+
			\frac{4a}{3\epsilon\phi_{\min}}
			\right)\ln\frac{4}{\delta} \ ,
			\\
			m & \ge \frac{128}{n^2\epsilon^2\phi_{\min}^2}\ln\frac{8}{\delta} \ .
		\end{split}
	\end{equation*}
	Then, there is an estimator using the measurement statistics from~\cref{fig_coh_qc} that satisfies \(
	\Pr(|1-\widehat{\phi}_{\cc}/\phi|\le \epsilon) \ge 1-\delta
	\).
\end{restatable}
\noindent The proofs of \Cref{thm:relMV,thm:relCC}
are relegated to \Cref{app_theorems}.

A remark on the measurement basis in \cref{fig_maj_qc,fig_coh_qc} is in order.
In practice, we perform measurements in the $Y$ basis for both circuits.
This choice is supported by the fact that $\Pr(0)=(1+\sin\phi)/2$ in the $Y$ basis is
more sensitive than $\Pr(0)=(1+\cos\phi)/2$ in the $X$ basis for small~$\phi$.
However, our proof of theoretical success guarantees
requires
measurements in the $X$ basis as well, 
hence our choice above 
to let $M$ denote either change of basis.

We conclude the section with a remark and a decision guarantee.

\begin{remark}
	\label{cor:queries}
	Let $q:=nm$ be the total number of unitary calls.
	The shot
	conditions of \Cref{thm:relMV,thm:relCC} read
	\begin{equation*}
		\begin{split}
			q_{\mv} & \ge            \frac{128}{\cos^2a\,\epsilon^2\phi_{\min}^2}\ln\frac{8}{\delta} \ , \\
			q_{\cc} & \ge
			\frac{128}{n\,\epsilon^2\phi_{\min}^2}\ln\frac{8}{\delta} \ .
		\end{split}
	\end{equation*}

\end{remark}

\begin{corollary}[Decision guarantee]
	\label{cor:sign}
	Assume the hypotheses of
	\Cref{thm:relMV} or \Cref{thm:relCC}.
	Since $\epsilon\in(0,1)$,
	the event $|1-\widehat\phi/\phi|\le\epsilon$,
	which occurs with probability at least $1-\delta$,
	implies
	that $\widehat\phi$ and $\phi$ have the same sign,
	hence that $\phi$ is correctly classified.
\end{corollary}

\section{Experimental tests}

\begin{figure*}[t]
	\centering
	\includegraphics[width=0.49\linewidth]{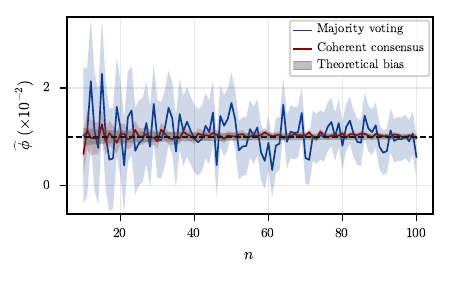}
	\hfill
	\includegraphics[width=0.49\linewidth]{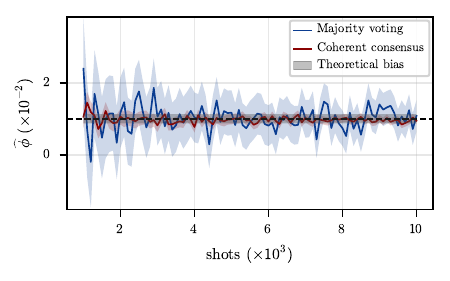}
	\caption{Estimated phases
		from the noiseless simulations
		of majority voting and coherent consensus circuits
		(\cref{fig_maj_qc,fig_coh_qc}),
		against the number of qubits (left,
		$m = 2\,000$)
		and against the number of shots (right,
		$n = 20$).
		The dashed line corresponds to the target $\phi = 10^{-2}$.
		Colored shaded bands represent the statistical error
		at $95\%$ confidence.
		Perturbations $\gamma_j$ are drawn from $\uniform(-a,a)$ with $a=5\times10^{-3}$.
		The gray shaded band represents
		the range of the theoretical bias
		at $95\%$ confidence.
		The absence of device noise is reflected
		in the lack of systematic bias,
		so that both estimators are accurate up to the theoretical bias.
		The confidence interval of coherent consensus
		is narrower than that of majority voting
		and shrinks faster
		as $n$ grows (left).
		The theoretical bias becomes comparable to the statistical error of coherent consensus for $n\gtrsim60$ (left) and $m\gtrsim6\times10^3$ (right).
		}
	\label{fig_sim_phase}
\end{figure*}

\label{sec_sim_and_exp}

\noindent
We deploy the coherent consensus and majority voting primitives on both noiseless simulators and quantum hardware to test their viability in realistic scenarios.
The following experiments target a phase $\phi = 10^{-2}$. Imperfect rotations are obtained by drawing $\gamma_j\sim\uniform(-a,a)$ with $a = 5 \times 10^{-3}$, where $\uniform(-a,a)$
is the uniform distribution
over $[-a,a]$.
Note that, by the central limit theorem, 
the average of the sampled $\gamma_j$ lies, with high probability, within a half-width of order $a/\sqrt{n}$ around zero.
This setting is designed to mimic the 
binary classification encoded in the sign
of an unknown, bounded,
small phase~$\phi$,
predicted using potentially imperfect gates $U_{\phi_j}$.
All measurements are performed in the $Y$ basis.

\subsection*{Simulation}
\noindent
We first demonstrate that coherent consensus outperforms majority voting in noiseless circuits.
We run simulations with numbers of qubits and shots typical of current and near-term devices. 
We show
the estimated phase in \cref{fig_sim_phase} and
its relative error in \cref{fig_sim_error}.
Both majority voting and coherent consensus are accurate,
meaning they do not display any bias beyond the theoretical bias of the imperfect gates.
However,
coherent consensus is more precise than majority voting,
as it displays a narrower confidence interval. 
The scaling laws in \mbox{\cref{stat_var_mv,stat_var_coh}}
are observable in the relative error.

The bounds derived in \Cref{thm:relMV,thm:relCC} are worst-case guarantees,
which hold for all settings satisfying the assumptions.
Being based on Chernoff-type inequalities, they are loose in practice.
For instance,
in the best-case scenario $\phi_{\min}=\phi_{\max}=\phi$,
with $\phi=10^{-3}$, $a=5\times10^{-4}$, $\epsilon=5\%$, and $\delta=5\%$,
both theorems require $n\gtrsim 1\,200$ qubits,
together with $m\gtrsim 2.1\times10^{8}$ shots for majority voting
and $m\gtrsim 1.7\times10^{5}$ shots for coherent consensus.
Yet coherent consensus already attains $\epsilon\lesssim5\%$
for every $n\gtrsim400$ at $m=10^{4}$ in \cref{fig_sim_error} (left),
and for every $m\gtrsim2\times10^{4}$ at $n=250$ (right),
that is, with about a third of the qubits and over an order of magnitude fewer shots than the bound requires (left), or about a fifth of the qubits and nearly an order of magnitude fewer shots (right).
Majority voting, for which the bound requires $2.1\times10^{8}$ shots,
would reach $\epsilon\lesssim5\%$ 
for every $n\gtrsim1.5\times10^{5}$ at $m=10^{4}$, or every $m\gtrsim6\times10^{6}$ at $n=250$, beyond the range of \cref{fig_sim_error}.
These results come from single runs and therefore do not establish the success probability $1-\delta$ guaranteed by the theorems,
but they show how conservative the bounds are.

\begin{figure*}[t]
	\centering
	\includegraphics[width=0.49\linewidth]{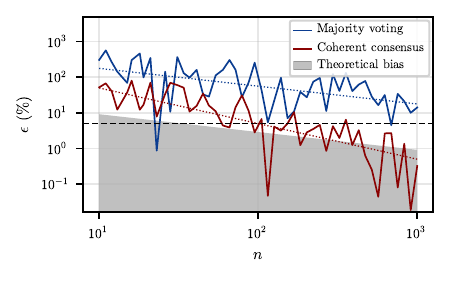}
	\hfill
	\includegraphics[width=0.49\linewidth]{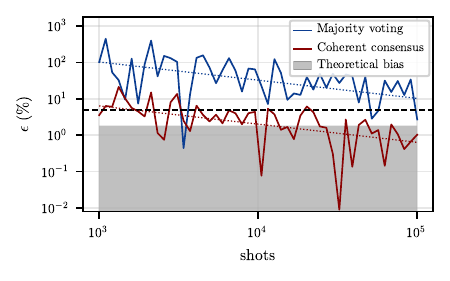}
	\caption{Relative errors
		from the simulations
		of majority voting and coherent consensus circuits
		(\cref{fig_maj_qc,fig_coh_qc}),
		against the number of qubits (left, $m = 10^4$)
		and against the number of shots (right, $n = 250$).
		The guide lines show the expected statistical-error scaling $\propto n^{-1/2}$
		for majority voting and $\propto n^{-1}$ for coherent consensus (left), and $\propto m^{-1/2}$
		for both majority voting and coherent consensus (right).
		The dashed line corresponds to
		the reference value $\epsilon = 5\%$.
		The target is
		scaled from $\phi = 10^{-2}$
		to $\phi = 10^{-3}$.
		Perturbations are correspondingly scaled such that 
		$\gamma_j$ are drawn from $\uniform(-a,a)$ with $a=5\times10^{-4}$.
		This scaling allows us to span $n=10,\ldots,1\,000$ while preserving $n\phi < \pi/2$, as required to obtain $\widehat{\phi}_{\cc}$.
		The gray shaded band represents
		the range of the theoretical bias
		at $95\%$ confidence over repeated draws.
		The absence of device noise is reflected
		in the lack of systematic bias,
		so that the relative error tends to zero as $n$ grows,
		and to the realized bias as $m$ grows at fixed $n$.
		Coherent consensus outperforms majority voting
		at nearly every point
		and reaches $\epsilon\lesssim 5\%$
		for
		$n \gtrsim 4\times10^2$ (left)
		and
		$m \gtrsim$ $2\times10^4$ (right).
		The theoretical bias becomes comparable to the statistical error of coherent consensus for $n\sim10^3$ (left) and $m\sim5\times10^4$ (right).
	}
	\label{fig_sim_error}
\end{figure*}

\subsection*{Hardware}
\noindent
To test the robustness of coherent consensus and majority voting on current noisy hardware, 
we run circuits on a 98-qubit trapped-ion quantum computer~\cite{Ransford25}.
Such devices~\cite{Bruzewicz19, FossFeig24}
pair long coherence times and high gate fidelities with slow gates and readout.
We estimate the target phase on \mbox{$n=20,\,50,\,80$} qubits
using \mbox{$m=500,\,1\,000,\,2\,000$} shots for both majority voting and coherent consensus, and we present the results in \cref{fig_exp_coverage} (left).
The two strategies are executed in separate campaigns with identical realized ensembles of phase gates.

A precision gain on hardware depends on whether the observed dispersion corresponds to shot noise.
We test this in \cref{fig_exp_ratio},
where we split each experiment into batches of consecutive shots.
The dispersion of majority voting and of coherent consensus at $n=20,50$ is consistent with shot noise.
For coherent consensus, the dispersion exceeds that predicted by shot noise by a factor of four at $n=80$ with $m=1\,000, 2\,000$, consistent with device drift during the run.

Moreover,~\cref{fig_exp_coverage} 
shows that device noise degrades the accuracy of both majority voting and coherent consensus.
Coherent consensus is particularly affected by two sources of error that impact the generation of the GHZ state.
First, the error rate of two-qubit gates is higher than that of single-qubit gates. 
This exclusively affects coherent consensus
and is unavoidable for entanglement generation. Second, this construction is also more sensitive to memory errors, that is, noise accumulation on idle ions during the computation.
We counteract these effects by compiling the GHZ preparation into logarithmic depth in $n$, which, compared with constant-depth preparation, reduces the number of two-qubit gates and limits idle time.
Experimental results show that coherent consensus is more accurate than majority voting for $n=20,50$ at every shot budget tested, given current hardware error rates,
while majority voting overtakes coherent consensus for $n = 80$ due to the loss of accuracy produced by systematic errors.
Note that the accuracy degradation of coherent consensus cannot be ascribed to the theoretical bias.

\subsection*{The cost
	of coherent consensus}
\noindent
In practice, the cost of acquiring one sample for coherent consensus on a quantum device is higher than that for majority voting, due to the resources needed to generate entanglement and measure the coherent sum of phases.
A fair comparison should take this into account.
Hardware quantum credits (HQCs) are a metric that quantifies the cost of running a circuit on the hardware we are using for the experiments~\cite{HQCs}.
They are a function of the number of one- and two-qubit gates as well as measurements and qubit initializations, each with its own cost and all weighted by the number of shots.
For a fixed HQC budget, the two-qubit gate overhead of coherent consensus has a direct impact on its overall shot budget.
This produces a larger constant factor in the cost of coherent consensus without changing the overall linear scaling.
The results are contained in \cref{fig_exp_coverage} (right), which shows that under these adversarial conditions coherent consensus is still more accurate than majority voting for $n=20,50$.
This makes coherent consensus preferable not only at equal shot budgets, but also at equal credit budgets, below the limit where systematic error dominates.

\begin{figure}[t]
	\centering
	\includegraphics[width=\linewidth]{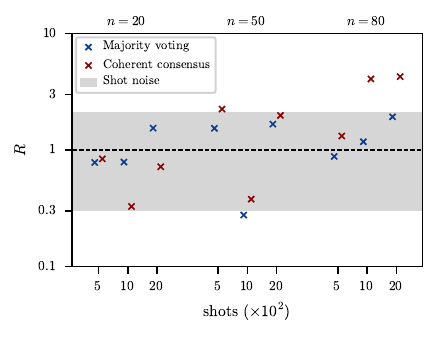}
	\caption{Dispersion ratio for the hardware runs at equal shot budgets shown in \cref{fig_exp_coverage} (left).
      Each run is split into ten batches of consecutive shots.
	  The dispersion ratio $R$ is the variance of the batch estimates $\widehat\phi$ divided by that predicted by shot noise.
      Independent shots with a common outcome probability give $R\approx1$ (dashed line).
      The gray shaded area is the $95\%$ shot-noise reference:
      under a normal approximation of the batch estimates, $9R$ approximately follows a $\chi^2$ distribution with nine degrees of freedom.
      Device noise correlated across shots, such as a drift during the run, raises $R$ above~$1$,
      while anticorrelated shots would lower it below the band.
      Majority voting lies within or just below the band ($8/9$ within).
      For $n=20,50$, coherent consensus lies within the band ($5/6$), with the remaining run just above.
      For $n=80$, coherent consensus lies well above the band ($2/3$).
      This excess batch dispersion is consistent with noise correlated across shots, 
	  but the ratio alone does not identify its source.}
	\label{fig_exp_ratio}
\end{figure}

\section{Conclusions}

\noindent In this work,
we introduced coherent consensus,
a quantum primitive
for the estimation and classification of phases with imperfect phase gates.
This primitive borrows from quantum metrology the use of entanglement for signal amplification.
Coherent consensus offers a $\sqrt{n}$ gain in precision over incoherent majority voting.
It is also well suited to parallelization on current quantum hardware,
which offers many qubits but limited coherence time. 
We established the advantage of coherent consensus over majority voting with rigorous success guarantees on the number of qubits $n$ and of shots (\Cref{thm:relMV,thm:relCC}).
This advantage holds for typical values of $n$ and typical shot budgets in classical simulations (\cref{fig_sim_phase,fig_sim_error}).
Furthermore,
we tested it on a state-of-the-art trapped-ion platform (\cref{fig_exp_ratio,fig_exp_coverage}).

The experimental results demonstrate that the precision advantage of coherent consensus holds on current quantum hardware. 
This gain in precision, 
however, 
does not translate into a gain in accuracy for large system sizes due to the error rates of the hardware implementation. 
Systematic implementation errors and theoretical modeling errors bias the estimates away from the target, even at high precision. 
By designing a scenario in which theoretical errors are under control, 
we isolate and study the effect of systematic quantum hardware errors on this task.
The results indicate that, 
although both strategies suffer from noise degradation, 
coherent consensus outperforms majority voting in accuracy for the tested qubit counts $n=20,50$, 
in all equal-shot comparisons and most equal-credit comparisons.

Single-qubit classifiers can be universal function approximators~\cite{PerezSalinas21}, 
as can single-hidden-layer neural networks~\cite{Hornik91}. 
Universal approximation alone does not establish a computational advantage.
To obtain a non-classical gain, classifiers should coordinate at the quantum level.
In the phase-estimation setting studied here, 
coherent consensus provides this coordination through entanglement, 
reducing sampling error relative to independent measurements. 
To exploit this primitive more broadly, 
further machine learning techniques must be encoded into the action of a phase gate on entangled probes.

Finally, we stress that coherent consensus is a ready-to-deploy primitive
for phase estimation and sign classification with imperfect phase gates.
Other standalone quantum techniques may follow similar steps to become part of a toolkit of primitives that can advance algorithms executed on current and future quantum devices.

\section*{Acknowledgments}

\noindent The authors thank Matthew Ho, Jun Yong, and Siong Thye for useful discussions. The authors also thank Irfan Khan for support with hardware deployment on Helios-1. 
This project is supported by the National Research Foundation, Singapore through the National Quantum Office, hosted in A*STAR, under its Advanced Quantum Computing Infrastructure Funding Initiative (S25Q8D9001) and its Centre for Quantum Technologies Funding Initiative (S24Q2d0009).

\section*{Data availability}

\noindent The data are available from the authors upon reasonable request.

\bibliographystyle{coherent_consensus_bibliographicstyle}
\bibliography{coherent_consensus_references}

\clearpage
\onecolumngrid
\appendix
\renewcommand{\theHequation}{appendix.\thesection.\arabic{equation}}
\section{Proofs of the theorems}\label[appendix]{app_theorems}

\noindent In this appendix, we prove \Cref{thm:relMV,thm:relCC}.
We begin by introducing some preliminary results.

\subsection{Preliminaries}

\label{section_test}

\begin{fact}[Bernstein inequality~\citeapp{Boucheron13}]\label{fact:bernstein-inequality}
    Let $X_1,\ldots,X_n$ be independent random variables with means \mbox{$\mu_j=\mathbb E[X_j]$}, and suppose there exists $R>0$ such that $|X_j-\mu_j|\le R$ almost surely for all~$j=1,\dots,n$.
    Define
    $
        S_n := \sum_{j=1}^n (X_j-\mu_j)
    $
    and
    $
        V := \sum_{j=1}^n \Var X_j \ .
    $
    Then, for all $t>0$, the one- and two-sided bounds hold:
    \begin{align*}
         & \Pr\!\big(S_n \ge t\big)
        \leq
        \exp\!\left(
        -\frac{t^2}{2V+\tfrac{2}{3}Rt}
        \right),                     \\
         & \Pr\!\big(|S_n|\ge t\big)
        \leq
        2\exp\!\left(
        -\frac{t^2}{2V+\tfrac{2}{3}Rt}
        \right) \ .
    \end{align*}
\end{fact}

\begin{fact}[Hoeffding inequality~\citeapp{Boucheron13}]\label{fact:hoeffding}
    Let $X_1,\dots,X_n$ be independent random variables with $X_j\in[0,1]$ and let $\widehat \mu:=\frac{1}{n}\sum_{j=1}^n X_j$ have mean $\mu$. Then, for every $u>0$,
    \begin{equation*}
        \Pr\bigl(|\widehat \mu-\mu|\ge u\bigr)\le 2\e^{-2nu^2} \ .
    \end{equation*}
\end{fact}

\begin{definition}[Principal value and two-argument arctangent]
    \label{def:atan2}
    The \emph{principal value} of an angle $\alpha\in\mathbb R$ is
    the unique representative
    $\operatorname{pv}(\alpha)\in(-\pi,\pi]$ satisfying
    \begin{equation*}
        \alpha-\operatorname{pv}(\alpha)\in2\pi\mathbb Z \ .
    \end{equation*}
    Pairs of angles
    differing by integer multiples of $2\pi$
    share the same principal value. For $(x,y)\ne(0,0)$, choose any $\alpha\in\mathbb R$ such that
    \begin{equation*}
        (\cos\alpha,\sin\alpha)
        =\frac{(x,y)}{\sqrt{x^2+y^2}} \ ,
    \end{equation*}
    and define
    \begin{equation*}
        \operatorname{atan2}(y,x):=\operatorname{pv}(\alpha) \ .
    \end{equation*}
    For $(x,y)=(0,0)$, we
    set $\operatorname{atan2}(0,0):=0$.
\end{definition}

\begin{fact}[Properties of the two-argument arctangent]
    \label{fact:atan2}
    The function in \Cref{def:atan2} has the following properties.
\begin{enumerate}
        \renewcommand{\labelenumi}{(\arabic{enumi})}
        \item \emph{Smoothness.} On
              $\mathbb R^2\setminus\{(x,0):x\le0\}$, the function
              $A(x,y):=\operatorname{atan2}(y,x)$ is smooth and
              \begin{equation*}
                  \nabla A(x,y)
                  =\frac{(-y,x)}{x^2+y^2},
                  \quad
                  \|\nabla A(x,y)\|_2
                  =\frac{1}{\sqrt{x^2+y^2}}.
              \end{equation*}

        \item \emph{Circular stability.} Let
              $\mathbf v=R\,(\cos\phi,\sin\phi)$ with $R>0$ and
              $\phi\in(-\pi,\pi)$, and let
              $\widehat{\mathbf v}=(\widehat x,\widehat y)$ satisfy
              $d:=\|\widehat{\mathbf v}-\mathbf v\|_2<R$.
              Then $\widehat{\mathbf v}\ne(0,0)$ and
              \begin{equation*}
                  \min_{k\in\mathbb Z}
                  \left|
                  \operatorname{atan2}(\widehat y,\widehat x)
                  -\phi+2\pi k
                  \right|
                  \le\arcsin\frac{d}{R} \ .
              \end{equation*}

        \item \emph{Principal-value stability.} Under the assumptions
              of circular stability, if 
              $\arcsin(d/R)<\pi-|\phi|$ also holds, then
              \begin{equation*}
                  \left|
                  \operatorname{atan2}(\widehat y,\widehat x)-\phi
                  \right|
                  \le\arcsin\frac{d}{R} \ .
              \end{equation*}
\end{enumerate}
\end{fact}

\begin{proof}
 Differentiating $A$ away from its branch cut gives the smoothness
    formulas. For the two stability properties, write
    \begin{equation*}
        \widehat{\mathbf v}
        =\rho\,(\cos(\phi+\vartheta),\sin(\phi+\vartheta)) \ ,
    \end{equation*}
    where $\rho>0$ and $\vartheta\in(-\pi,\pi]$ is the smaller
    signed angle from the direction of $\mathbf v$ to that of $\widehat{\mathbf v}$. Then
    \begin{equation*}
        d^2
        =R^2+\rho^2-2R\rho\cos\vartheta
        =(\rho-R\cos\vartheta)^2+R^2\sin^2\vartheta \ .
    \end{equation*}
    Since $d<R$, it follows that
    $\cos\vartheta>0$, hence
    $|\vartheta|<\pi/2$. 
    It also follows that
    $|\vartheta|\le\arcsin(d/R)$,
    thus proving circular stability.
    The additional phase-margin
    condition $\arcsin(d/R)<\pi-|\phi|$ implies 
    $|\phi|+|\vartheta|<\pi$, which
    prevents the angular displacement from crossing the
    branch cut of $\operatorname{atan2}$, thus proving principal-value stability.
\end{proof}

\begin{fact}[Bound on the absolute error for single-qubit phase estimation]
    \label{fact:phase-estimation-absolute}
    Let $\alpha\in(-\pi,\pi)$ be a phase.
    Denote the absolute-error target by $\epsilon\in(0,\pi/2)$
    and the failure probability by $\delta\in(0,1)$,
    and assume
    \begin{equation*}
        |\alpha|+\epsilon<\pi \ .
    \end{equation*}
    Let the even number of shots $m$ satisfy
    \begin{equation*}
        m\ge
        \frac{32}{\epsilon^2}
        \ln\frac{4}{\delta}\ .
    \end{equation*}
    Then there is an estimator $\widehat{\alpha}$ satisfying
    \(
        \Pr\bigl(|\widehat{\alpha}-\alpha|\le\epsilon\bigr)
        \ge 1-\delta \ .
    \)
\end{fact}

\begin{proof}
    Use $m/2$ shots for measurements in the $X$ basis and 
    $m/2$ shots for measurements in the $Y$ basis.
    The probabilities of the outcome $0$ are
    \begin{equation*}
        p_X=\frac{1+\cos\alpha}{2}
        \ , \quad
        p_Y=\frac{1+\sin\alpha}{2} \ .
    \end{equation*}
    Let
    $\widehat p_X$ and $\widehat p_Y$ be the corresponding empirical
    frequencies of outcome $0$.
    Each empirical frequency is formed from $m/2$ independent
    Bernoulli outcomes.
    Define
    \begin{equation*}
        \widehat\alpha
        :=
        \operatorname{atan2}
        \bigl(2\widehat p_Y-1,\ 2\widehat p_X-1\bigr) \ .
    \end{equation*}
    Set
    $u:=\epsilon/(4\sqrt2)$ and let $\mathcal E$ be the event
    \begin{equation*}
        |\widehat p_X-p_X|\le u
        \quad\text{and}\quad
        |\widehat p_Y-p_Y|\le u \ .
    \end{equation*}
    The Hoeffding inequality and a union bound for the complementary event give
    \begin{equation*}
        \Pr(\mathcal E^{\mathsf c})
        \le4\exp(-mu^2)
        =4\exp(-m\epsilon^2\!/32)
        \le\delta \ .
    \end{equation*}
    On $\mathcal E$, define
    \begin{equation*}
        \widehat{\mathbf v}
        :=
        (2\widehat p_X-1,2\widehat p_Y-1)
        =\|\widehat{\mathbf v}\|_2\,(\cos\widehat{\alpha},\sin\widehat{\alpha}) \ ,
        \quad
        \mathbf v:=(\cos\alpha,\sin\alpha) 
        =(2 p_X-1,2 p_Y-1)\ .
    \end{equation*}
    Then \(
    \|\widehat{\mathbf v}-\mathbf v\|_2
    \le2\sqrt2\,u
    =\epsilon/2.
    \)
    Since $\epsilon\in(0,\pi/2)$,
    it follows that
    $\arcsin(\epsilon/2)<\epsilon$.  Together with
    $\epsilon<\pi-|\alpha|$, principal-value stability in
    \Cref{fact:atan2} with $R = 1$ yields
    $
    |\widehat\alpha-\alpha|<\epsilon
    $
    on $\mathcal E$, proving the claim.
\end{proof}

\relativeMVlearning*

\begin{proof}
	Set
	$
		\eta:={\epsilon\phi_{\min}/2}
	$.
	The phase-range assumption implies
	$0<\eta<\pi/2$ and
	\(
	\pi-|\phi|>2\eta.
	\)
	First consider the random phase realizations. In coordinates
	rotated by $-\phi$, define
	\begin{equation*}
		C_n:=\frac{1}{n}\sum_{j=1}^n\cos\gamma_j,
		\quad
		S_n:=\frac{1}{n}\sum_{j=1}^n\sin\gamma_j,
		\quad
		R_n:=\sqrt{C_n^2+S_n^2}.
	\end{equation*}
	Rotating by $-\phi$ expresses each gate phase relative to the target, turning the unit vector $(\cos\phi_j,\sin\phi_j)$ into $(\cos\gamma_j,\sin\gamma_j)$; thus $(C_n,S_n)$ is the average of the rotated vectors, and its angle is the deviation of the averaged gate direction from $\phi$. We use these coordinates because this deviation is then controlled by $S_n$, a mean of independent, zero-mean, bounded variables to which the Bernstein inequality applies.
	Since $|\gamma_j|\le a<\pi/2$,
	$C_n\ge \cos a>0$
	and $R_n\ge \cos a$.
	Symmetry gives $\mathbb E[\sin\gamma_j]=0$, while
	$\Var(\sin\gamma_j)\le \sigma^2$ and
	$|\sin\gamma_j|\le\sin a$.
	Let $\mathcal F:=\{|S_n|\le\cos a\tan\eta\}$ be the event that the sampled $\gamma_j$ satisfy this inequality. If $a=\sigma=0$, then $\gamma_j=0$ almost surely for every $j$, so $\Pr(\mathcal F)=1$. Otherwise, $a>0$, and the following concentration calculation applies.
	Applying the Bernstein inequality
	(\Cref{fact:bernstein-inequality})
	yields
	\begin{equation*}
		\Pr\bigl(|S_n|>\cos a\tan\eta\bigr)
		\le
		2\exp\left(
		-\frac{n \cos^2 a\tan^2\eta}{2\sigma^2+\frac{2}3\sin a \cos a\tan\eta}
		\right) \ .
	\end{equation*}
	Using $\tan\eta\ge\eta$ and the assumed lower bound on $n$ gives
	\begin{equation*}
		\Pr\bigl(|S_n|\le \cos a \tan\eta\bigr)
		\ge1-\frac{\delta}{2} \ .
	\end{equation*}
	This probability bound also holds in the zero-noise case.
	For every phase realization, let
	$\vartheta_n:=\operatorname{atan2}(S_n,C_n)$.  Since
	$C_n\ge \cos a$,
	on $\mathcal F$ it follows that
	$
		|\vartheta_n|\le\eta.
	$
	Thus the noiseless combined phase
	$\phi_n:=\phi+\vartheta_n$ (defined as a real, unwrapped angle for every realization) differs from $\phi$ by at most $\eta$ on $\mathcal F$. It remains to control measurement noise.

	As in~\Cref{fact:phase-estimation-absolute}, we use $m/2$ shots for measurements in the $X$ basis and
	$m/2$ shots for measurements in the $Y$ basis.
	Each basis is measured for $m/2$ shots, giving $nm/2$
	independent, but not necessarily identically distributed, Bernoulli
	outcomes for fixed $\gamma_j$.
	Let
	$\widehat p_{X,j}$ and $\widehat p_{Y,j}$ be the empirical
	frequencies of outcome $0$ for the $j$th qubit.
	Their means are
	\begin{equation*}
		p_{X,j}:=\frac{1+\cos(\phi+\gamma_j)}{2},
		\qquad
		p_{Y,j}:=\frac{1+\sin(\phi+\gamma_j)}{2}.
	\end{equation*}
	Define the probabilities pooled over all qubits by
	\begin{equation*}
			p_{X}  :=\frac{1}{n}\sum_{j=1}^n p_{X,j}=\frac{1+R_n\cos\phi_n}{2}  \ , \quad
			p_{Y}  :=\frac{1}{n}\sum_{j=1}^n p_{Y,j}=\frac{1+R_n\sin\phi_n}{2} \ .
	\end{equation*}
	To obtain the closed forms, note that $R_n\ge\cos a>0$, so \Cref{def:atan2} gives $C_n=R_n\cos\vartheta_n$ and $S_n=R_n\sin\vartheta_n$. The angle-addition formulas then give
	\begin{equation*}
		\begin{aligned}
			\frac{1}{n}\sum_{j=1}^n\cos(\phi+\gamma_j)
			 & =C_n\cos\phi-S_n\sin\phi
			=R_n(\cos\vartheta_n\cos\phi-\sin\vartheta_n\sin\phi)
			=R_n\cos\phi_n \ , \\
			\quad
			\frac{1}{n}\sum_{j=1}^n\sin(\phi+\gamma_j)
			 & =C_n\sin\phi+S_n\cos\phi
			=R_n(\cos\vartheta_n\sin\phi+\sin\vartheta_n\cos\phi)
			=R_n\sin\phi_n \ ,
		\end{aligned}
	\end{equation*}
	where the last equalities use $\phi_n=\phi+\vartheta_n$. Averaging $p_{X,j}$ and $p_{Y,j}$ over $j$ and substituting these sums yields the expressions for $p_X$ and $p_Y$.
	For the majority voting estimator, we first pool these frequencies:
	\begin{equation*}
		\widehat p_X:=\frac{1}{n}\sum_{j=1}^n\widehat p_{X,j} \ ,
		\quad
		\widehat p_Y:=\frac{1}{n}\sum_{j=1}^n\widehat p_{Y,j} \ .
	\end{equation*}
	Define
	\begin{equation*}
		\widehat{\phi}_{\mv}:=\operatorname{atan2}(2\widehat p_Y-1,\,2\widehat p_X-1).
	\end{equation*}
	The centered vector of the pooled probabilities is
	\begin{equation*}
		\mathbf v_n
		:=(2p_{X}-1,2p_{Y}-1)
		=R_n(\cos\phi_n,\sin\phi_n) \ .
	\end{equation*}
	Set
	$
		u
		:=({\cos a\sin\eta})/({2\sqrt2})
	$
	and
	let $\mathcal E$ be the event
	\begin{equation*}
		|\widehat p_X-p_X|\le u
		\quad\text{and}\quad
		|\widehat p_Y-p_Y|\le u \ ,
	\end{equation*}
	conditional on
	the phases $\gamma_j$.
	The Hoeffding inequality and a union bound give
	\begin{equation*}
		\Pr(
		\mathcal E^{\mathsf c}
		)
		\le4\exp(-nm{u}^2)=4\exp(-nm\cos^2 a\sin^2\eta/8) \ .
	\end{equation*}
	Since $\sin\eta>\eta/2=\epsilon\phi_{\min}/4$, the stated shot
	bound makes the right-hand side at most $\delta/2$.
	This bound does not depend on the values of the $\gamma_j$, so it also holds unconditionally: $\Pr(\mathcal E^{\mathsf c})\le\delta/2$.
	On the intersection of the measurement event $\mathcal E$ and the phase event $\mathcal F$, define
	\begin{equation*}
		\widehat{\mathbf v}
		:=(2\widehat p_{X}-1,2\widehat p_{Y}-1)
		=\|\widehat{\mathbf v}\|_2\,(\cos\widehat{\phi}_{\mv},\sin\widehat{\phi}_{\mv})
		\ .
	\end{equation*}%
	Then
	\(
	\|\widehat{\mathbf v}-\mathbf v_n\|_2
	\le2\sqrt2\,u
	=\cos a\sin\eta
	\le R_n\sin\eta<R_n \ .
	\)
	Moreover,
	$\pi-|\phi_n|>\eta$ on $\mathcal F$, since $|\phi_n|\le|\phi|+\eta$ and $\pi-|\phi|>2\eta$.  Applying \Cref{fact:atan2} with reference
	radius $R_n$ therefore gives
	\begin{equation*}
		|\widehat{\phi}_{\mv}-\phi_n|\le\eta \ .
	\end{equation*}%
	On $\mathcal E\cap\mathcal F$, the triangle inequality gives
	\begin{displaymath}
		|\widehat{\phi}_{\mv}-\phi|
		\le
		|\widehat{\phi}_{\mv}-\phi_n|
		+|\phi_n-\phi|
		\le2\eta
		=\epsilon\phi_{\min}
		\le\epsilon|\phi| \ .
	\end{displaymath}
	A union bound proves that this occurs with probability at least
	$1-\delta$, because $\Pr(\mathcal F^{\mathsf c})\le\delta/2$ and $\Pr(\mathcal E^{\mathsf c})\le\delta/2$.
\end{proof}

\relativeCClearning*

\begin{proof}
	Define the real, unwrapped sum $\theta:=\sum_{j=1}^n\phi_j=n\phi+\sum_{j=1}^n\gamma_j$. If $a=\sigma=0$, the phase error is zero and the phase event below has probability one without invoking the Bernstein inequality. Otherwise, $a>0$, and the following concentration argument applies.
	Apply the Bernstein inequality (\Cref{fact:bernstein-inequality}) to
	\begin{equation*}
		S_n:=\sum_{j=1}^n \gamma_j
	\end{equation*}
	with $R=a$ and $V=n\sigma^2$.
	For any $u>0$,
	\begin{equation*}
		\begin{split}
			\Pr\!\big(|\theta-n\phi|\ge nu\big)
			=\Pr\!\big(|S_n|\ge nu\big)
			\le
			2\exp\!\left(
			-\frac{n\,u^2}{2\sigma^2+\frac{2}{3}au}
			\right) \ .
		\end{split}
	\end{equation*}
	Set $u=\epsilon|\phi|/2$ and define
	\begin{equation*}
		g(z):=
		\frac{z^2}{8\sigma^2+\frac{4}3a\epsilon z} \ ,
	\end{equation*}
	for $z>0$. For nonzero noise parameters,
	\begin{equation*}
		g'(z)
		=
		\frac{
			16\sigma^2z+\frac{4}3a\epsilon z^2
		}{
			\left(8\sigma^2+\frac{4}3a\epsilon z\right)^2
		}
		>0 \ .
	\end{equation*}
	(If $a=\sigma=0$, then $\theta=n\phi$ deterministically.)
	Thus $|\phi|\ge\phi_{\min}$ implies
	$g(|\phi|)\ge g(\phi_{\min})$.  This is why the unknown
	$|\phi|$ in the Bernstein exponent can be replaced by its known
	lower bound when fixing the qubit budget.
	The stated bound on
	$n$ consequently gives
	\begin{equation*}
		\Pr(\mathcal E_{\mathrm{phase}})
		\ge1-\frac{\delta}{2} \ ,
		\quad
		\mathcal E_{\mathrm{phase}}
		:=
		\left\{
		|\theta-n\phi|
		\le\frac{\epsilon n|\phi|}{2}
		\right\}.
	\end{equation*}

	Set
	$
		\eta:=n\epsilon\phi_{\min}/2
	$.
	The assumption on the range of $\phi$ gives
	$
		|\theta|
		\le n(\phi_{\max}+a)
		<\pi/(1+\epsilon)
	$.
	It then follows that
	\begin{equation*}
		0<\eta
		<\frac{\epsilon\pi}{2(1+\epsilon)}
		<\frac{\pi}{2} \ ,
		\quad
		|\theta|+\eta
		<\frac{1+\epsilon/2}{1+\epsilon}\pi
		<\pi \ .
	\end{equation*}
	For fixed $\gamma_j$, the readout probabilities are $q_X=(1+\cos\theta)/2$ and $q_Y=(1+\sin\theta)/2$.
	Let $\widehat q_X$ and $\widehat q_Y$ be the empirical
	frequencies of
	outcome $0$ in the coherent consensus readout, and define
	\begin{equation*}
		\widehat\theta
		:=
		\operatorname{atan2}
		\bigl(2\widehat q_Y-1,\ 2\widehat q_X-1\bigr) \ ,
		\quad
		\widehat\phi_{\cc}:=\frac{\widehat\theta}{n} \ .
	\end{equation*}
	The stated shot bound is equivalent to
	\begin{equation*}
		m
		\ge\frac{32}{\eta^2}
		\ln\frac{8}{\delta}
		=
		\frac{128}{n^2\epsilon^2\phi_{\min}^2}
		\ln\frac{8}{\delta}\ ,
	\end{equation*}
	so applying \Cref{fact:phase-estimation-absolute} conditionally on
	the realized phases, with $\alpha=\theta$, absolute error
	$\eta$, and failure probability $\delta/2$, guarantees
	\begin{equation*}
		\Pr(\mathcal E_{\mathrm{shot}})\ge1-\frac{\delta}{2} \ , \quad
		\mathcal E_{\mathrm{shot}}
		:=
		\left\{
		|\widehat\theta-\theta|
		\leq \frac{n\epsilon\phi_{\min}}{2}
		\right\} \ .
	\end{equation*}
	The additional factor $1/n$ relative to majority voting comes from estimating
	the amplified phase $\theta$, approximately equal to $n\phi$.
	On $\mathcal E_{\mathrm{phase}}\cap\mathcal E_{\mathrm{shot}}$,
	the triangle
	inequality yields
	\begin{equation*}
		\begin{split}
			|\phi-\widehat\phi_{\cc}|
			\le \frac{|n\phi-\theta|}{n} + \frac{|\theta-\widehat{\theta}|}{n}
			\le \frac{\epsilon\,|\phi|}{2} + \frac{\epsilon\,\phi_{\min}}{2}
			\le \epsilon\,|\phi| \ ,
		\end{split}
	\end{equation*}
	using $\phi_{\min}\le|\phi|$.  
    A union bound gives
	\(
	\Pr(
	\mathcal E_{\mathrm{phase}}
	\cap
	\mathcal E_{\mathrm{shot}}
	)
	\ge1-\delta,
	\)
	proving the claim.
\end{proof}

    \noindent The proofs
    use the
    two-argument
    arctangent
    (\Cref{def:atan2})
    and its
    properties~(\Cref{fact:atan2}).
    The
    shot budget
    is
    accordingly
    evenly
    split
    between
    measurements
    in the $X$
    and $Y$ bases
    to estimate
    $\cos\alpha$
    and
    $\sin\alpha$,
    respectively,
    where
    $\alpha$
    is
    the
    phase
    of
    interest for
    majority voting or coherent consensus.
    Knowledge of
    the bounds on $\phi$
    may
    inform
    alternative
    shot allocations
    in practice.
    For instance,
    in~\Cref{sec_sim_and_exp},
    $\phi$ is
    known to be small,
    and
    all
    the
    measurements
    are carried out in the
    $Y$ basis.

\makeatletter
\begingroup
  \def\NAT@bibsetnum#1{\setlength{\topsep}{\z@}\NATx@bibsetnum{\ref{LastBibItem-app}}}%
  \def\endthebibliography{\auto@bib@innerbib\edef\@currentlabel{\arabic{NAT@ctr}}%
    \label{LastBibItem-app}\endNAT@thebibliography\aftergroup\auto@bib@empty}%
\providecommand{\href}[2]{#2}\begingroup\raggedright\endgroup
\endgroup
\makeatother

\end{document}